\documentclass[12pt,twoside]{article}
\usepackage{xcolor}
\usepackage{url}
\usepackage{enumitem}
\usepackage{amsmath}
\usepackage{amsthm}
\usepackage{amssymb}
\usepackage{setspace}
\PassOptionsToPackage{normalem}{ulem}
\usepackage{ulem}
\usepackage{microtype}

\usepackage[unicode=true,
 bookmarks=true,bookmarksnumbered=false,bookmarksopen=false,
 breaklinks=false,pdfborder={0 0 1},backref=false,colorlinks=false]
 {hyperref}
\hypersetup{pdftitle={\csname input@path\endcsname ;  \jobname ; \today}}

\makeatletter

\providecommand{\Nathanael}{\ifmmode\ERROR\else{Nathana\"{e}l}\xspace\fi}
\usepackage{xspace}
\usepackage{xfrac}
\usepackage{dsfont}

\usepackage[all]{xy}

\usepackage{xspace}
\usepackage{xfrac}
\usepackage{dsfont}

\usepackage{changepage}   % for the adjustwidth environment
\DeclareTextSymbolDefault{\textquotedbl}{T1}
\RequirePackage{thmtools, thm-restate}
\RequirePackage{currfile}

\theoremstyle{plain}
\newtheorem{thm}{\protect\theoremname}
\theoremstyle{definition}
\newtheorem{defn}[thm]{\protect\definitionname}
\theoremstyle{plain}

\theoremstyle{remark}
\newtheorem{claim}[thm]{\protect\claimname}
\newenvironment{proof-sketch}[0]%Defined in MYtheorems.module
	{\begin{proof}[\textbf{Proof sketch}]}
	{\end{proof}}
\theoremstyle{plain}
\newtheorem{lem}[thm]{\protect\lemmaname}
\theoremstyle{plain}
\newtheorem{innerLem}{\protect\lemmaname}

\theoremstyle{plain}
\newtheorem{conjecture}[thm]{\protect\conjecturename}

\newcommand{\MyLyxThmFN}[1]{\hspace{-6pt}\footnote{#1}}
\newcommand{\MyLyxThmNewline}[1]{~\par\nopagebreak\ignorespaces}

\usepackage[cm]{fullpage}
\usepackage{microtype}
\DisableLigatures{encoding = *, family = *}

\usepackage{ifthen}
\setlist[itemize,enumerate]{nosep}
\usepackage{hyphenat}%\hyphenation{group-ma-ni-pu-la-tion}

\makeatother

\usepackage{expl3}
\ExplSyntaxOn
\int_zero_new:N \g__prg_map_int 
\ExplSyntaxOff
\usepackage{tikz}\usetikzlibrary{calc,decorations.pathreplacing,calligraphy}

\providecommand{\claimname}{Claim}
\providecommand{\conjecturename}{Conjecture}
\providecommand{\corollaryname}{Corollary}
\providecommand{\definitionname}{Definition}
\providecommand{\lemmaname}{Lemma}
\providecommand{\theoremname}{Theorem}

\begin{document}

\global\long\def\Z{\mathbb{Z}}
\global\long\def\N{\mathbb{N}}
\global\long\def\C{\mathbb{C}}
\global\long\def\Q{\mathbb{Q}}

\global\long\def\GR{G}
\global\long\def\GRe{E}
\global\long\def\GRvNEW{\mathcal{V}}
\global\long\def\GRe{E}

\global\long\def\GRvVerts{V}
\global\long\def\GRzVerts{Z}

\newcommand{\ZVlineGraphTEXT}{$\GRzVerts\GRvVerts$-line graph\xspace}
\newcommand{\ZVlineGraphsTEXT}{$\GRzVerts\GRvVerts$-line graphs\xspace}
\newcommand{\ZVopartTEXT}{$\GRzVerts\GRvVerts$-ordered partition\xspace}
\newcommand{\ZVopartsTEXT}{$\GRzVerts\GRvVerts$-ordered partitions\xspace}

\global\long\def\Zvertex{\GRzVerts\text{-vertex}}
\global\long\def\Zvertices{\GRzVerts\text{-vertices}}

\global\long\def\Vvertex{\GRvVerts\text{-vertex}}
\global\long\def\Vvertices{\GRvVerts\text{-vertices}}

\global\long\def\SubGraph#1{\GRvVerts_{#1}}

\global\long\def\ViSubGraph{\SubGraph i\text{-subgraph}}
\global\long\def\ViSubGraphs{\SubGraph i\text{-subgraphs}}

\global\long\def\Clique#1{K_{#1}}
\global\long\def\Cycle#1{C_{#1}}

\global\long\def\WRTaPartitionInner#1#2#3#4{#1=#2\cup\left(#3_{1}\DisjCup\cdots\DisjCup#3_{#4}\right)}

\global\long\def\WRTaPartition{\WRTaPartitionInner{\GRvNEW}{\GRzVerts}{\GRvVerts}k}

\global\long\def\XtoTheStar#1{\bigcup_{t\geqslant0}#1^{t}}
\global\long\def\VtoTheStar{\XtoTheStar{\GRvNEW}}

\global\long\def\LOC{x}
\global\long\def\LOCof#1{\LOC_{#1}}
\global\long\def\LOCvecGen#1{\boldsymbol{#1}}
\global\long\def\LOCvec{\LOCvecGen{\LOC}}

\global\long\def\REPORTset{\boldsymbol{A}}

\global\long\def\DISTof#1#2{d\left(#1,#2\right)}

\global\long\def\NeighOf#1{N\left(#1\right)}
\global\long\def\BallOf#1#2{B\left(#1,#2\right)}

\global\long\def\NeighOfInZ#1{\NeighOf{#1}\cap\GRzVerts}
\global\long\def\BallOfCapZ#1#2{\BallOf{#1}{#2}\cap\GRzVerts}

\global\long\def\RootSymbol{\mathcal{R}}
\global\long\def\RootOfSubGraph#1{\RootSymbol\left(\SubGraph{#1}\right)}

\global\long\def\MECH{F}
\global\long\def\MECHof#1{\MECH\left(#1\right)}

\global\long\def\MECHi#1{\MECH_{#1}}
\global\long\def\MECHiOf#1#2{\MECHi{#1}\left(#2\right)}

\global\long\def\MECHz{G}
\global\long\def\MECHzOf#1{\MECHz\left(#1\right)}

\global\long\def\PO{PO}
\global\long\def\POof#1{\PO\left(#1\right)}

\global\long\def\TheMECH{\MECH^{\star}}
\global\long\def\TheMECHof#1{\TheMECH\left(#1\right)}

\global\long\def\xPrime#1{#1^{\prime}}
\global\long\def\xDPrime#1{#1^{\second}}
\global\long\def\xTPrime#1{#1^{\third}}
\global\long\def\xFPrime#1{#1^{\fourth}}

\makeatletter
\global\long\def\rmnum#1{\romannumeral#1}
\global\long\def\Rmnum#1{\uppercase\expandafter{\romannumeral#1\relax}}
\makeatother

\global\long\def\DisjCup{\mathop{\dot{\cup}} }
\global\long\def\bigDisjCup{\mathop{\dot{\bigcup}} }
\global\long\def\SetSt{\;\middle\vert\;}
\global\long\def\sizeof#1{\left|#1\right|}
\global\long\def\st{\text{ s.t. }}

\global\long\def\MyMathop#1{\mathop{\mathrm{#1}}}
\global\long\def\argmax{\MyMathop{argmax}}
\global\long\def\argmin{\MyMathop{argmin}}

\newcommand{\mathify}[1]{\ifmmode{#1}\else\mbox{\ensuremath{#1}}\fi} 
\newcommand{\texify}[1] {\ifmmode{\text{#1}}\else{#1}\fi}
\newcommand{\XXth}[2]{%
	\ifmmode {#1^{\underline{#2}}}%
	\else {#1\textsuperscript{\underline{#2}}\xspace}\fi%
}
\newcommand{\Nst}[1]{\mathify{\XXth{#1}{st}}\xspace}
\newcommand{\Nnd}[1]{\mathify{\XXth{#1}{nd}}\xspace}
\newcommand{\Nrd}[1]{\mathify{\XXth{#1}{rd}}\xspace}
\newcommand{\Nth}[1]{%
	\ifcase#1\relax%
	0%\ErrorXthCannotBeZero%
	\or\Nst{1}%
	\or\Nnd{2}%
	\or\Nrd{3}%
	\else\mathify{\Xth{#1}}%
	\fi\xspace}
\global\long\def\Xth#1{\XXth{#1}{th}\xspace}
\global\long\def\ith{\Xth i\xspace}
\global\long\def\jth{\Xth j\xspace}

\global\long\def\EDGE#1{\ar@{-}[#1]}
\global\long\def\EDGEwithDist#1#2{\ar_{#2}@{-}[#1]}
\global\long\def\longEDGE#1#2{\ar_{#2}@{~}[#1]}
\global\long\def\EDGEcurved#1#2{\ar@{-}@/^{#2}/[#1]}
\global\long\def\EDGEdot#1{\ar@{.}[#1]}
\global\long\def\VERTEX{\boldsymbol{\blacklozenge}}
\global\long\def\Zcirc{{\color{purple}\boldsymbol{\bigcirc}}}
\global\long\def\VERTEXwithName#1{*++[o][F]{#1}}

\newcommand{\GeneralExampleVRZ}[5]{%ViaResize
	%c/b/t
	%C=3emR=3em
	%\textwidth
	%!
	%content
	\begin{array}[#1]{@{}c@{}}%
		\resizebox{#3}{#4}{\xymatrix#2{#5}}%
\end{array}}%

\global\long\def\ThePropSymbol{{\color{cyan!80!black}\left(\bigstar\right)}}\newcommand{\ExampleBiClique}[1]{\ExampleBiCliqueInner{\VERTEX}{c}{}{#1}{!}}%
\newcommand{\ExampleBiCliqueAsZV}[1]{\ExampleBiCliqueInner{\Zcirc}{c}{}{#1}{!}}%
\newcommand{\ExampleBiCliqueInner}[5]{\GeneralExampleVRZ{#2}{#3}{#4}{#5}{	\EDGEdot{r}&
	\VERTEX\EDGE{d}\EDGE{drrr}\EDGE{drrrrrr}\EDGE{drrrrrrrrr}\EDGE{drrrrrrrrrrrr}&&&
	\VERTEX\EDGE{dlll}\EDGE{d}\EDGE{drrr}\EDGE{drrrrrr}\EDGE{drrrrrrrrr}&&&
	\VERTEX\EDGE{dllllll}\EDGE{dlll}\EDGE{d}\EDGE{drrr}\EDGE{drrrrrr}&&&
	\VERTEX\EDGE{dlllllllll}\EDGE{dllllll}\EDGE{dlll}\EDGE{d}\EDGE{drrr}&&&
	\VERTEX\EDGE{dllllllllllll}\EDGE{dlllllllll}\EDGE{dllllll}\EDGE{dlll}\EDGE{d}&
	\EDGEdot{l}\\
	\EDGEdot{r}&#1&&&#1&&&#1&&&#1&&&#1&\EDGEdot{l}
}}%

\newcommand{\ExampleABuniformZV}[1]{\ExampleABuniformZVinner{c}{}{#1}{!}}
\newcommand{\ExampleABuniformZVinner}[4]{\GeneralExampleVRZ{#1}{#2}{#3}{#4}{%	&&&&\VERTEX\EDGE{dl}\EDGE{d}\EDGE{d}&
	&&&\VERTEX\EDGE{dl}\EDGE{d}&
	&\VERTEX\EDGE{dl}\EDGE{d}&
	&\VERTEX\EDGE{dl}\EDGE{d}&
	&\VERTEX\EDGE{dl}\EDGE{d}&
	&\VERTEX\EDGE{dl}\EDGE{d}&
	\\
%	Z:&\cdots&\Zcirc&\Zcirc&\Zcirc&\Zcirc&
	\cdots&\Zcirc&\Zcirc&\Zcirc&\Zcirc&
\Zcirc&\Zcirc&\Zcirc&\Zcirc&\Zcirc&
\Zcirc&\Zcirc&\cdots\\
%	&&\VERTEX\EDGE{u}\EDGE{ur}&
	&\VERTEX\EDGE{u}\EDGE{ur}&
	&\VERTEX\EDGE{u}\EDGE{ur}&
	&\VERTEX\EDGE{u}\EDGE{ur}&
	&\VERTEX\EDGE{u}\EDGE{ur}&
	&\VERTEX\EDGE{u}\EDGE{ur}&
}}

\newcommand{\ExampleBBuniformFourBoxesZVinner}[4]{\GeneralExampleVRZ{#1}{#2}{#3}{#4}{	&&\VERTEX\EDGE{dl}\EDGE{dr}&&
		\VERTEX\EDGE{dl}\EDGE{dr}&&		\VERTEX\EDGE{dl}\EDGE{dr}&&		\VERTEX\EDGE{dl}\EDGE{dr}&&
	\\
	\EDGEdot{r}&\Zcirc&&\Zcirc		&&\Zcirc	&&\Zcirc&&\Zcirc&\EDGEdot{l}\\
	&&\VERTEX\EDGE{ul}\EDGE{ur}&&
		\VERTEX\EDGE{ul}\EDGE{ur}&&		\VERTEX\EDGE{ul}\EDGE{ur}&&		\VERTEX\EDGE{ul}\EDGE{ur}&&
}}

\newcommand{\ExampleBBuniformThreeBoxesZVinner}[4]{\GeneralExampleVRZ{#1}{#2}{#3}{#4}{	&&\VERTEX\EDGE{dl}\EDGE{dr}&&
		\VERTEX\EDGE{dl}\EDGE{dr}&&		\VERTEX\EDGE{dl}\EDGE{dr}&&
	\\
	\EDGEdot{r}&\Zcirc&&\Zcirc	&&\Zcirc&&\Zcirc&\EDGEdot{l}\\
	&&\VERTEX\EDGE{ul}\EDGE{ur}&&
		\VERTEX\EDGE{ul}\EDGE{ur}&&		\VERTEX\EDGE{ul}\EDGE{ur}&&
}}

\newcommand{\ExampleBBuniformZVdecorated}[1]{\ExampleBBuniformZVdecoratedinner{c}{}{#1}{!}{2}}
\newcommand{\ExampleBBuniformZVdecoratedinner}[5]{%
\newcommand{\MAINedge} {\EDGE{#5,-#5}\EDGE{#5,#5}}%
\newcommand{\MAINedgeB}{\EDGE{-#5,-#5}\EDGE{-#5,#5}}%
\GeneralExampleVRZ{#1}{#2}{#3}{#4}{&&  &\VERTEX\EDGE{rd}\EDGE{rdd}\EDGE{ld}\EDGE{ldd}\EDGE{ddd}
&& && && &&
\VERTEX\EDGE{r}\EDGE{rd}\EDGE{rdd}\EDGE{rddd}&\VERTEX\EDGE{ld}\EDGE{ldd}\EDGE{lddd}\\
&&  \VERTEX\EDGE{rr}\EDGE{d}\EDGE{drr}\EDGE{ddr}&&\VERTEX\EDGE{d}\EDGE{dll}\EDGE{ddl}
& && && &&
	\VERTEX\EDGE{r}\EDGE{rd}\EDGE{rdd}&\VERTEX\EDGE{ld}\EDGE{ldd}\\
&&  \VERTEX\EDGE{rr}\EDGE{dr}&&\VERTEX\EDGE{dl}
& && && &&
	\VERTEX\EDGE{r}\EDGE{rd}&\VERTEX\EDGE{ld}\\
 & && \VERTEX\MAINedge && &&\VERTEX\MAINedge && &&\VERTEX\MAINedge\EDGE{r}&\VERTEX  & &&\VERTEX\MAINedge
\\ \\
%\cdots 
	& \Zcirc && && \Zcirc && && \Zcirc && && \Zcirc && && \Zcirc & 
%\cdots
\\ \\
 & && \VERTEX\MAINedgeB && &&\VERTEX\MAINedgeB && &&\VERTEX\MAINedgeB && &&\VERTEX\MAINedgeB\\
 & && \VERTEX\EDGE{u}  && &\VERTEX\EDGE{ur}&\VERTEX\EDGE{u}&\VERTEX\EDGE{ul}
}}

\newcommand{\ExampleBBuniformZVdecoratedinnerB}[5]{%
\newcommand{\MAINedge} {\EDGE{#5,-#5}\EDGE{#5,#5}}%
\newcommand{\MAINedgeB}{\EDGE{-#5,-#5}\EDGE{-#5,#5}}%
\GeneralExampleVRZ{#1}{#2}{#3}{#4}{&&  &\VERTEX\EDGE{rd}\EDGE{rdd}\EDGE{ld}\EDGE{ldd}\EDGE{ddd}
&& && && &&
\VERTEX\EDGE{r}\EDGE{rd}\EDGE{rdd}\EDGE{rddd}&\VERTEX\EDGE{ld}\EDGE{ldd}\EDGE{lddd}\\
&&  \VERTEX\EDGE{rr}\EDGE{d}\EDGE{drr}\EDGE{ddr}&&\VERTEX\EDGE{d}\EDGE{dll}\EDGE{ddl}
& && && &&
	\VERTEX\EDGE{r}\EDGE{rd}\EDGE{rdd}&\VERTEX\EDGE{ld}\EDGE{ldd}\\
&&  \VERTEX\EDGE{rr}\EDGE{dr}&&\VERTEX\EDGE{dl}
& && && &&
	\VERTEX\EDGE{r}\EDGE{rd}&\VERTEX\EDGE{ld}\\
 & && \VERTEX\MAINedge && &&\VERTEX\MAINedge && &&\VERTEX\MAINedge\EDGE{r}&\VERTEX  & &&\VERTEX\MAINedge
\\ \\
%\cdots 
	& \Zcirc && && \Zcirc\EDGE{rrrr} && && \Zcirc && && \Zcirc && && \Zcirc & 
%\cdots
\\ \\
 & && \VERTEX\MAINedgeB && &&\VERTEX\MAINedgeB && &&\VERTEX\MAINedgeB && &\VERTEX\EDGE{uul}&\VERTEX\MAINedgeB\\
 & && \VERTEX\EDGE{u}  && &\VERTEX\EDGE{ur}&\VERTEX\EDGE{u}&\VERTEX\EDGE{ul}&&&&\VERTEX\EDGE{uuur}&\VERTEX\EDGE{ur}&\VERTEX\EDGE{u}&\VERTEX\EDGE{ul}
}}

\newcommand{\ExampleBBuniformZVdecoratedinnerC}[5]{%
\newcommand{\MAINedge} {\EDGE{#5,-#5}\EDGE{#5,#5}}%
\newcommand{\MAINedgeB}{\EDGE{-#5,-#5}\EDGE{-#5,#5}}%
\GeneralExampleVRZ{#1}{#2}{#3}{#4}{&&  &\VERTEX\EDGE{rd}\EDGE{rdd}\EDGE{ld}\EDGE{ldd}\EDGE{ddd}
&& && && &&
\VERTEX\EDGE{r}\EDGE{rd}\EDGE{rdd}\EDGE{rddd}&\VERTEX\EDGE{ld}\EDGE{ldd}\EDGE{lddd}
	&&&& \VERTEX\EDGE{rr}\EDGE{rddddd} \EDGE{d}\EDGE{drr} \EDGE{dd}\EDGE{dddrr} \EDGE{dd}\EDGE{dddrr}  \EDGEcurved{dd}{-1em}\EDGEcurved{ddd}{-2em}
		&&\VERTEX\EDGE{lddddd} \EDGE{d}\EDGE{dll} \EDGE{dd}\EDGE{ddll} \EDGE{ddd}\EDGE{dddll}  \EDGEcurved{dd}{1em}\EDGEcurved{ddd}{2em}
\\
&&  \VERTEX\EDGE{rr}\EDGE{d}\EDGE{drr}\EDGE{ddr}&&\VERTEX\EDGE{d}\EDGE{dll}\EDGE{ddl}
& && && &&
	\VERTEX\EDGE{r}\EDGE{rd}\EDGE{rdd}&\VERTEX\EDGE{ld}\EDGE{ldd}
	&&&& \VERTEX\EDGE{rr}\EDGE{rdddd} \EDGE{d}\EDGE{drr} \EDGE{dd}\EDGE{ddrr}  \EDGEcurved{dd}{-1em}
		&&\VERTEX\EDGE{ldddd} \EDGE{d}\EDGE{dll} \EDGE{dd}\EDGE{ddll}  \EDGEcurved{dd}{1em}\\
&&  \VERTEX\EDGE{rr}\EDGE{dr}&&\VERTEX\EDGE{dl}
& && && &&
	\VERTEX\EDGE{r}\EDGE{rd}&\VERTEX\EDGE{ld}
	&&&& \VERTEX\EDGE{rr}\EDGE{rddd} \EDGE{d}\EDGE{drr}
		&&\VERTEX\EDGE{lddd} \EDGE{d}\EDGE{dll}
\\
 & && \VERTEX\MAINedge && &&\VERTEX\MAINedge && &&\VERTEX\MAINedge\EDGE{r}&\VERTEX  & &&\VERTEX\MAINedge
	&\VERTEX\EDGE{rr}\EDGE{rdd}&&\VERTEX\EDGE{ldd}
\\ \\
%\cdots 
	& \Zcirc && && \Zcirc\EDGE{rrrr} && && \Zcirc && && \Zcirc && && \Zcirc & 
%\cdots
\\ \\
 & && \VERTEX\MAINedgeB && &&\VERTEX\MAINedgeB && &&\VERTEX\MAINedgeB && &\VERTEX\EDGE{uul}&\VERTEX\MAINedgeB\\
 & && \VERTEX\EDGE{u}  && &\VERTEX\EDGE{ur}&\VERTEX\EDGE{u}&\VERTEX\EDGE{ul}&&&&\VERTEX\EDGE{uuur}&\VERTEX\EDGE{ur}&\VERTEX\EDGE{u}&\VERTEX\EDGE{ul}
}}

\newcommand{\ExampleACuniformFiveBoxesZVinner}[4]{\GeneralExampleVRZ{#1}{#2}{#3}{#4}{	&&\VERTEX\EDGE{dl}\EDGE{d}\EDGE{dr}&&
		\VERTEX\EDGE{dl}\EDGE{d}\EDGE{dr}&&
		\VERTEX\EDGE{dl}\EDGE{d}
	\\
	\EDGEdot{r}&
		\Zcirc&\Zcirc&\Zcirc&\Zcirc&\Zcirc&\Zcirc
	&\EDGEdot{l}\\
	&\VERTEX\EDGE{u}\EDGE{ur}&
	&\VERTEX\EDGE{ul}\EDGE{u}\EDGE{ur}&
	&\VERTEX\EDGE{ul}\EDGE{u}\EDGE{ur}
}}

\newcommand{\ExampleACuniformFourBoxesZVinner}[4]{\GeneralExampleVRZ{#1}{#2}{#3}{#4}{	&&\VERTEX\EDGE{dl}\EDGE{d}\EDGE{dr}&&
		\VERTEX\EDGE{dl}\EDGE{d}\EDGE{dr}&&
	\\
	\EDGEdot{r}&
		\Zcirc&\Zcirc&\Zcirc&\Zcirc&\Zcirc
	&\EDGEdot{l}\\
	&\VERTEX\EDGE{u}\EDGE{ur}&
	&\VERTEX\EDGE{ul}\EDGE{u}\EDGE{ur}&
	&\VERTEX\EDGE{ul}\EDGE{u}
}}

\newcommand{\ExampleBoxesSevenBoxesInner}[4]{\GeneralExampleVRZ{#1}{#2}{#3}{#4}{\EDGEdot{r} & \VERTEX\EDGE r\EDGE d & \VERTEX\EDGE r\EDGE d & 
\VERTEX\EDGE r\EDGE d & \VERTEX\EDGE r\EDGE d & \VERTEX\EDGE r\EDGE d& \VERTEX\EDGE r\EDGE d & \VERTEX\EDGE r\EDGE d & \VERTEX\EDGE d & \EDGEdot{l}\\
\EDGEdot{r} & \VERTEX\EDGE r & \VERTEX\EDGE r & \VERTEX\EDGE r & \VERTEX\EDGE r & \VERTEX\EDGE r & \VERTEX\EDGE r & \VERTEX\EDGE r 
& \VERTEX & \EDGEdot{l}
}}%

\newcommand{\ExampleBoxesFiveBoxesInner}[4]{\GeneralExampleVRZ{#1}{#2}{#3}{#4}{\EDGEdot{r} & \VERTEX\EDGE r\EDGE d & \VERTEX\EDGE r\EDGE d & 
\VERTEX\EDGE r\EDGE d & \VERTEX\EDGE r\EDGE d & \VERTEX\EDGE r\EDGE d & \VERTEX\EDGE d & \EDGEdot{l}\\
\EDGEdot{r} & \VERTEX\EDGE r & \VERTEX\EDGE r & \VERTEX\EDGE r & \VERTEX\EDGE r & \VERTEX\EDGE r 
& \VERTEX & \EDGEdot{l}
}}%

\newcommand{\ExampleBoxesFourBoxesInner}[4]{\GeneralExampleVRZ{#1}{#2}{#3}{#4}{\EDGEdot{r} & \VERTEX\EDGE r\EDGE d & \VERTEX\EDGE r\EDGE d & 
\VERTEX\EDGE r\EDGE d & \VERTEX\EDGE r\EDGE d & \VERTEX\EDGE r\EDGE d & \EDGEdot{l}\\
\EDGEdot{r} & \VERTEX\EDGE r & \VERTEX\EDGE r & \VERTEX\EDGE r & \VERTEX\EDGE r 
& \VERTEX & \EDGEdot{l}
}}%

\newcommand{\ExampleThreeBoxes}[1]{\ExampleThreeBoxesinner{c}{}{#1}{!}}
\newcommand{\ExampleThreeBoxesinner}[4]{\GeneralExampleVRZ{#1}{#2}{#3}{#4}{	\VERTEX\EDGE{r}\EDGE{d}&\VERTEX\EDGE{d}\\
	\VERTEX\EDGE{r}\EDGE{d}&\VERTEX\EDGE{r}\EDGE{d}&\VERTEX\EDGE{d}\\
	\VERTEX\EDGE{r}&\VERTEX\EDGE{r}&\VERTEX
}}%

\newcommand{\ExampleF}[1]{\ExampleFinner{c}{}{#1}{!}}
\newcommand{\ExampleFinner}[4]{\GeneralExampleVRZ{#1}{#2}{#3}{#4}{%	&&&&\VERTEX
%		\EDGE{dl}\EDGE{dlll}
%		\EDGE{dr}\EDGE{drrr}
%	\\
%	Z:&\Zcirc&&\Zcirc&&\Zcirc&&\Zcirc&&&&\\
%	&&\VERTEX\EDGE{ul}\EDGE{ur}
%	&&\VERTEX\EDGE{ul}\EDGE{ur}&
%	&\VERTEX\EDGE{ul}\EDGE{ur}\\
	&&&\VERTEX
		\EDGE{dl}\EDGE{dlll}
		\EDGE{dr}\EDGE{drrr}
	\\
	\Zcirc&&\Zcirc&&\Zcirc&&\Zcirc\\
	&\VERTEX\EDGE{ul}\EDGE{ur}
	&&\VERTEX\EDGE{ul}\EDGE{ur}&
	&\VERTEX\EDGE{ul}\EDGE{ur}
}}%

\newcommand{\ZigZagGrid}[1]{\ZigZagGridinner{c}{}{#1}{!}}
\newcommand{\ZigZagGridinner}[4]{\GeneralExampleVRZ{#1}{#2}{#3}{#4}{&  & \VERTEX &  & \VERTEX &  & \VERTEX &  & \VERTEX &  & \EDGEdot{dl}\\
 & \VERTEX\EDGE{ur} &  & \VERTEX\EDGE{ul}\EDGE{ur} &  & \VERTEX\EDGE{ul}\EDGE{ur} &  & \VERTEX\EDGE{ul}\EDGE{ur} &  & \VERTEX\EDGE{ul} &  & \EDGEdot{dl}\\
\EDGEdot{ur} &  & \VERTEX\EDGE{ul}\EDGE{ur} &  & \VERTEX\EDGE{ul}\EDGE{ur} &  & \VERTEX\EDGE{ul}\EDGE{ur} &  & \VERTEX\EDGE{ul}\EDGE{ur} &  & \VERTEX\EDGE{ul}\\
 & \EDGEdot{ur} &  & \VERTEX\EDGE{ul}\EDGE{ur} &  & \VERTEX\EDGE{ul}\EDGE{ur} &  & \VERTEX\EDGE{ul}\EDGE{ur} &  & \VERTEX\EDGE{ul}\EDGE{ur}
}}

\newcommand{\TwoCycle}[3]{\TwoCycleInner{c}{}{#1}{!}{#2}{#3}}
\newcommand{\TwoCycleInner}[6]{\GeneralExampleVRZ{#1}{#2}{#3}{#4}{#5\EDGE{r}  &  #6}}%

\newcommand{\ThreeCycle}[4]{\ThreeCycleInner{c}{}{#1}{!}{#2}{#3}{#4}}
\newcommand{\ThreeCycleInner}[7]{\GeneralExampleVRZ{#1}{#2}{#3}{#4}{	  &  #5\EDGE{ddr}\EDGE{ddl}\\
 \\
 #6\EDGE{rr}  &   &  #7
}}%

\newcommand{\FourCycle}[5]{\FourCycleInner{c}{}{#1}{!}{#2}{#3}{#4}{#5}}
\newcommand{\FourCycleInner}[8]{\GeneralExampleVRZ{#1}{#2}{#3}{#4}{  &  #5\EDGE{dr}\EDGE{dl}\\
 #8  &   &  #6\\
  &  #7\EDGE{ur}\EDGE{ul} 
}}%

\newcommand{\FiveCycle}[6]{\FiveCycleInner{c}{}{#1}{!}{#2}{#3}{#4}{#5}{#6}}
\newcommand{\FiveCycleInner}[9]{\GeneralExampleVRZ{#1}{#2}{#3}{#4}{&   &  #5\EDGE{drr}\EDGE{dll}\\
 #9  &   &   &   &  #6\\
  &  #8\EDGE{lu}\EDGE{rr}  &   &  #7\EDGE{ru} 
}}%

\newcommand{\FiveCycleDecorated}[4]{\GeneralExampleVRZ{#1}{#2}{#3}{#4}{&&		
&&   &&  \VERTEX\EDGE{ddrrrr}\EDGE{ddllll} && && 
&&& \VERTEX\EDGE{dll}\\
&\VERTEX\EDGE{dd}\EDGE{rd}&
&&   &&   &&   &&  
& \VERTEX\EDGE{ld}\\
\VERTEX\EDGE{rr}\EDGE{ru}\EDGE{rd}&&		 
\VERTEX  &&   &&   &&   &&  \VERTEX
 			&&& \VERTEX\EDGE{llu}\\
&\VERTEX\EDGE{ru}&
&&   &&   &&   &&  
&& \VERTEX\EDGE{llu}\\
&&		  
&&  \VERTEX\EDGE{lluu}\EDGE{rrrr}  &&   &&  \VERTEX\EDGE{rruu} 
}}

\newcommand{\CsixLabeledinner}[4]{\GeneralExampleVRZ{#1}{#2}{#3}{#4}{	& \VERTEXwithName 0\EDGE{rd}\\
	\VERTEXwithName 5\EDGE{ru} &  & \VERTEXwithName 1\EDGE d\\
	\VERTEXwithName 4\EDGE u &  & \VERTEXwithName 2\EDGE{ld}\\
	 & \VERTEXwithName 3\EDGE{lu}
}}%

\newcommand{\ExampleCounterExampleWeightedGraph}[1][]{%
	\begin{array}{c}%
		\xymatrix@C=0pt#1{%
			& \VERTEXwithName v\EDGEwithDist{dl}1\\
			\VERTEXwithName{z_{\ell}} &  & \VERTEXwithName{z_{r}}\EDGEwithDist{ul}{10}
		}\\
		 \begin{array}{l}
			 \GRvVerts=\left\{  v\right\}  \\
			 \GRzVerts=\left\{  z_{\ell},z_{r}\right\}  
		\end{array} \end{array} }

\title{Manipulation-resistant facility location mechanisms for \ZVlineGraphsTEXT}

\author{Ilan Nehama\thanks{Corresponding author\protect \\
\smallskip{}
\protect \\
We would like to thank Kentaro Yahiro and \Nathanael Barrot for
their comments which helped us to improve the presentation of this
work.  This work was partially supported by JSPS KAKENHI Grant Numbers
JP17H00761 and JP17H04695,  JST Strategic International Collaborative
Research Program, SICORP, and Israel Science Foundation Grant 1626/18.}\quad{}Taiki Todo\quad{}Makoto Yokoo\\
}
\maketitle
\begin{abstract}

In many real-life scenarios, a group of agents needs to agree on a
common action, e.g., on the location for a public facility, while
there is some consistency between their preferences, e.g., all preferences
are derived from a common metric space.  The \emph{facility location}
problem models such scenarios and it is a well-studied problem in
social choice.  We study mechanisms for facility location on unweighted
undirected graphs, which are resistant to manipulations (\emph{strategy-proof},
\emph{abstention-proof}, and \emph{false-name-proof}) by both individuals
and coalitions and are efficient (\emph{Pareto optimal}).  We define
a family of graphs, \emph{\ZVlineGraphsTEXT}, and show a general
facility location mechanism for these graphs which satisfies all these
desired properties.  Moreover, we show that this mechanism can be
computed in polynomial time, the mechanism is anonymous, and it can
equivalently be defined as the first Pareto optimal location according
to some predefined order.

Our main result, the \ZVlineGraphsTEXT family and the mechanism we
present for it, unifies the few current works in the literature of
false-name-proof facility location on discrete graphs, including all
the preliminary (unpublished) works we are aware of.  Finally, we
discuss some generalizations and  limitations of our result for problems
of facility location on other structures.
\end{abstract}

\newpage{}

\tableofcontents{}

\newpage{}

\section{Introduction}
Reaching an agreement could be hard. The seminal works of Gibbard~\cite{Gibbard1973} and Satterthwaite~\cite{Satterthwaite1975}
show that one cannot devise a general procedure for aggregating the
preferences of strategic agents to a single outcome, besides trivial
procedures that a-priori ignore all agents except one (that is, the
outcome is based on the preference of a predefined agent) or a-priori
rule out all outcomes except two (that is, regardless of the agents'
preferences, the outcome is one of two predefined outcomes). The problem
is that  agents might act strategically aiming to get an outcome
which they prefer, so there might be scenarios in which for any profile
of actions (a possible agreement) at least one of the agents will
prefer changing his action.   Note that while we refer to a \emph{procedure}
and later to a \emph{mechanism}, this impossibility is not technical
but conceptual. We identify a procedure with the conceptual mapping
induced by the procedure from the opinions of the agents to an agreement,
while the procedure itself could be complex and abstract, e.g., to
have several rounds or include a deliberation process between the
agents (cheap-talk). For simplicity of terms, we refer to the \emph{direct
mechanism} which implements this mapping. That is, we think of an
exogenous entity, the \emph{designer}, who receives as input the opinions
of the agents and returns as output the aggregated decision.\footnote{For the properties we study in this work, this assumption does not
hurt the generality, as according to the revelation principle~\cite{Myerson79},
any general procedure is equivalent (w.r.t. the properties we study)
to such a direct mechanism.}

 But in many natural scenarios,  it is exogenously given that the
preferences satisfy some additional rationality property, i.e., the
mechanism should not be defined for any profile of preferences, giving
rise to mechanisms that are not prone to the above drawbacks. Two
prominent examples are \emph{VCG mechanisms} and \emph{generalized-median
mechanisms}. VCG mechanisms~\cite{Vickrey1960,clarke1971,Roberts1979,groves1973}
are the mechanisms which are resistant to manipulations like the ones
described above for scenarios in which the agents' preferences are
quasi-linear with respect to money~\cite[Def.~3.b.7]{Mas-Colell1995},
and monetary transfers are allowed .  The second example, \emph{Generalized-median
mechanisms}, do not include monetary transfers and have more of an
ordinal flavor. Generalized-median mechanisms~\cite{Moulin1980}
are the mechanisms which are resistant to manipulations like above
when it is known that the preferences are \emph{single-peaked} w.r.t.
the real line~\cite{Black1948}. That is, the outcomes are \emph{locations}
on the real line, each agent has a unique optimal location, $\ell^{\star}$,
and her preference over the locations to the right of $\ell^{\star}$
is derived by the proximity to $\ell^{\star}$, and similarly for
the locations to the left of $\ell^{\star}$. For example, in the
Euclidean single-peaked case, the preferences for all agents are minimizing
the distance to their respective optimal locations.

\subsection*{The facility location problem}

A natural generalization of the second scenario is the \emph{facility
location} problem. In this problem,  we are given a metric space over the outcomes (that is, a distance
function between outcomes) and it is assumed that the preference of
each of the agents is defined by the distance to her optimal outcome:
An agent with an optimal outcome $\ell^{\star}$ prefers outcome $a$
over outcome $b$ if and only if  $a$ is closer to $\ell^{\star}$
than $b$. For ease of presentation, throughout this paper we assume
that there are finitely many agents and finitely many locations. In
this case, a natural way to represent the common metric space is using
a weighted undirected graph. That is, having a vertex (location) for
each outcome and weighted edges between vertices s.t. the distance
between any two outcomes is equal to the distance between the two
respective vertices.   Roughly speaking, given such a graph one
seeks to find a mechanism that on one hand will not a-priori ignore
some of the voters or rule-out some of the locations, and on the other
hand will be resistant to manipulations of the agents.    Facility location problems, and moreover facility location problems
for complex combinatorial structures, model many real-life scenarios
of group decision making in which it is natural to assume some homogeneity
between the different agents' preferences (e.g., an additional rationality
assumption). These examples include not only locating a common facility,
like a school, a bus-stop, or a library, but also more general agreement
scenarios with a common metric, e.g., partition of a common budget
to several tasks, committee selection, and group decision making with
a multi-dimensional criteria. Following the common facility problem,
we sometimes refer to the outcome of the mechanism as \emph{the facility}.

 In this work, we seek mechanisms which satisfy the following desired
properties:

%\vspace{-1em}

\paragraph*{Anonymity:}

The mechanism should not a-priori ignore agents and moreover it should
treat them equally in the following strong sense. The mechanism should
be a function of the agents' votes (which we also refer to as \emph{ballots})
but not their identities. Formally, the outcome of the mechanism should
be invariant to voters exchanging votes, i.e., to a permutation of
the ballots. In practice, most voting systems satisfy this property
by first accumulating the different (physical) ballots, thus losing
the voters' identities, and next applying the mechanism on the identity-less
ballots.

%\vspace{-1em}

\paragraph*{Onto:\protect\footnote{In the social choice literature~\cite{Aleskerov200295}, this property
is referred to as \emph{Citizen sovereignty} or \emph{Non-Imposition}.}}

The mechanism should not a-priori rule-out a location, and each location
should be an outcome of some profile. Formally, the mapping to a facility
location should be an onto function.  Moreover, the mechanism should
respect the preferences of the agents and aim to optimize the aggregated
welfare of the agents. 

%\vspace{-1em}

\paragraph*{Pareto optimality:}

 The mechanism should not return a location $\ell$ if there exists
another location $\xPrime{\ell}$ s.t. switching from $\ell$ to $\xPrime{\ell}$
will benefit one of the agents (move the facility closer to her) while
not hurting any of the other agents. In particular, if there exists
a unique location which is unanimously most-preferred by all agents,
then it must be the outcome. Note that any reasonable (monotone) notion
of aggregated welfare optimization entails Pareto optimality.  (Note that it is unreasonable to require that all locations are treated
equally  due the inherent asymmetry induced by the graph.)

%\vspace{-1em}

\paragraph*{Strategy-proofness:}

An agent should not be able to change the outcome to a location she
strictly prefers by reporting a location different than her true location.

%\vspace{-1em}

\paragraph*{Abstention-proofness:\protect\footnote{In the voting literature (e.g., \cite{DBLP:conf/wine/Conitzer08,Moulin1988,Fishburn1983})
this property is also referred to as \textbf{voluntary participation}
and the \textbf{no-show paradox}. This property is also equivalent
to \textbf{individual-rationality} which takes a different point of
view of mechanism design.}}

An agent should not be able to change the outcome to a location she
strictly prefers by not casting a ballot.

%\vspace{-1em}

\paragraph*{False-name-proofness:}

An agent should not be able to change the outcome to a location she
strictly prefers by casting more than one ballot.

False-name-manipulations received less attention in the classic social
choice literature, since in most voting scenarios there exists a central
authority that can enforce a `one person, one vote' principle (but
cannot enforce participation or sincere voting).  In contrast, many
of the voting and aggregation scenarios nowadays are run in a distributed
manner on some network and include virtual identities or avatars,
which can be easily generated, so a manipulation of an agent pretending
to represent many voters is eminent.

%\vspace{-1em}

\paragraph*{Resistance to group manipulations:}

We also consider generalizations of the above three properties dealing
with manipulations of coalitions of agents. We define the \emph{preference
of a coalition} as the unanimous preference of its members, that is,
a coalition $C$ weakly prefers an outcome $a$ over an outcome $b$
if all the members of $C$ weakly prefer $a$ over $b$, and require
that a coalition should not be able to change the outcome to a location
it strictly prefers\footnote{Hence, $C$ strictly prefers $a$ over $b$ if $\boldsymbol{\left(\rmnum 1\right)}$
all the members of $C$ weakly prefer $a$ over $b$ ($C$ weakly
prefers $a$ over $b$), and $\boldsymbol{\left(\rmnum 2\right)}$
at least one member of $C$ strictly prefers $a$ over $b$ ($C$
does not weakly prefer $b$ over $a$).} by its members casting insincere ballots, abstaining, or casting
more than one ballot. We note that for onto mechanisms this property
entails Pareto optimality. Nevertheless, we prefer to think of Pareto
optimality apart from this property due to the different motivations.

\subsection*{Our contribution}

Besides the work of Todo et al.~\cite{DBLP:conf/atal/TodoIY11},
who characterized the false-name-proof mechanisms for facility location
on the continuous line and on continuous trees, we are not aware of
other works dealing with characterizing false-name-proof mechanisms
on a graph. Moreover, as far as we know, a false-name-proof mechanism
is known to the community only for very few simple graphs, and  the
current knowledge is still highly preliminary. (When starting to work
on this problem, we initially devised mechanisms for few of the examples
we describe below - cycles, cliques, and the $2\times n$ grid. We
are not aware of any other previously-known positive results besides
these graphs or small perturbations of them.)

In this paper we present a family of unweighted undirected graphs,
which we name \ZVlineGraphsTEXT, and show a general mechanism for
facility location over these graphs which satisfies the desired properties.
To the best of our knowledge, this is the first work to show a general
false-name-proof mechanism for a general family of graphs. Our mechanism
for the \ZVlineGraphsTEXT family unifies the few mechanisms that
are known and induces mechanisms for many other graphs. The mechanism
is Pareto optimal and in particular satisfies citizen sovereignty;
It is anonymous, so in particular no agent is ignored; But on the
other hand, it is resistant to all the above manipulations. 

Roughly speaking, in a \ZVlineGraphTEXT there are two types of locations
$\GRzVerts$ and $\GRvVerts$ (and we refer to them as $\Zvertices$
and $\Vvertices$, respectively), and the facility is `commonly' (except
if all agents unanimously agree differently) located on a $\Zvertex$.
For instance, the $\Zvertices$ could represent commercial locations
for locating a public mall, or a set of status-quo outcomes.

\global\long\def\EXAMPLEocINLINEgrpahHeight{.15em}

\global\long\def\EXAMPLEocExampleABuniformZV{\ExampleABuniformZVinner{c}{}{9.5cm}{!}}

\global\long\def\EXAMPLEocExampleBoxes{\ExampleBoxesinner{c}{}{!}{.2em}}

\global\long\def\EXAMPLEocExampleF{\ExampleFinner{c}{}{!}{\EXAMPLEocINLINEgrpahHeight}}

\global\long\def\EXAMPLEocExampleThreeBoxes{\ExampleThreeBoxesinner{c}{}{!}{\EXAMPLEocINLINEgrpahHeight}}

For example, consider the following family of graphs (which is a sub-family
of \ZVlineGraphsTEXT and captures the  gist of our mechanism). Let
$\GR=\left\langle \GRvNEW,\GRe\right\rangle $ be a bipartite unweighted
undirected graph with vertex set $\GRvNEW$ and edge set $\GRe$.
That is, there exists a partition of the vertices $\GRvNEW=\GRvVerts\DisjCup\GRzVerts$
s.t. there are no edges between $\Vvertices$ and no edges between
$\Zvertices$. In addition, we require that $\boldsymbol{\left(a\right)}$
there exists a predefined order over the $\Zvertices$, which we refer
to as left-to-right order, and that $\boldsymbol{\left(b\right)}$
any of the $\Vvertices$ is connected to an interval (according to
the order) of $\Zvertices$. Similarly to the single-peaked consistency
case~\cite{Black1948}, one can  think of this constraint as a homogeneity
constraint over the agents' preferences.  Our mechanism for such
graphs: 
\begin{itemize}
\item[$\blacktriangleright$]  The mechanism returns the leftmost Pareto optimal $\Zvertex$,\footnote{That is, there exists no other location $\ell$ in the graph s.t.
switching the outcome to $\ell$ benefits one of the agents while
not hurting any of the other agents.} if one exists.
\item[$\blacktriangleright$]  If no location in $\GRzVerts$ is Pareto optimal, then necessarily
all agents voted for the same location, and the mechanism returns
this location.
\end{itemize}
For example, bi-cliques (full bipartite graphs) can be represented
as a \ZVlineGraphTEXT in which each $\Vvertex$ is connected to all
the $\Zvertices$ as follows (and we use below $\Zcirc$ for $\Zvertices$
and $\VERTEX$ for $\Vvertices$):\vspace{-1.2em}

\[
\ExampleBiCliqueInner{\Zcirc}{c}{@R=3em}{.38\paperwidth}{!}
\]

\noindent{}Our mechanism for this case:
\begin{itemize}
\item[$\blacktriangleright$]  If all agents voted unanimously for the same location, the mechanism
returns this location.
\item[$\blacktriangleright$]  If all agents voted for $\Vvertices$, the mechanism returns the
leftmost $\Zvertex$. 
\item[$\blacktriangleright$]  Otherwise, the mechanism returns the leftmost $\Zvertex$ that was
voted for.
\end{itemize}
Notice that in this case the order over the $\Zvertices$ is arbitrary
(as well as the choice of one of the sides to be the $\Zvertices$)
in the sense that it is not derived from the graph but a parameter
of the mechanism. For instance, the order might represent the social
norm of the society. 

A second example is the discrete line graph, which can be represented
as a \ZVlineGraphTEXT in which every two consecutive $\Zvertices$
are connected by a unique $\Vvertex$, $\ExampleABuniformZVinner{c}{@R=1.5ex@C=1.5ex}{!}{.3em}$.
  In particular, we show strategy-proof, false-name-proof, Pareto
optimal mechanisms which are far from \emph{generalized-median mechanisms}
(for instance, in the common case the output of the mechanism belongs
to a subset consisting of only half of the locations), in contrary
to the characterization of these mechanisms for the continuous line~\cite[Thm.~2]{DBLP:conf/atal/TodoIY11}.
 Dokow et al.~\cite[Thm.~3.4]{Dokow2012} characterized the strategy-proof
mechanisms for the discrete line as a superset of generalized-median
mechanisms, hence we get a strict subset of their characterization
(and actually a small fraction of their characterization) due to requiring
also false-name-proofness.

Two simple  graphs that are generalizations of (the \ZVlineGraphTEXT
representation of) the discrete line graph are   $\ExampleBBuniformThreeBoxesZVinner{c}{@R=1.5ex@C=1.5ex}{!}{.3em}$,
in which every two consecutive $\Zvertices$ are connected by two
$\Vvertices$, and the $2\times n$ grid  $\ExampleBoxesFourBoxesInner{c}{@R=2ex@C=2ex}{!}{.3em}$
which can be represented as a \ZVlineGraphTEXT in which every three
consecutive $\Zvertices$ are connected by a unique $\Vvertex$, i.e.,
 $\ExampleACuniformFourBoxesZVinner{c}{@R=1em}{3.5cm}{!}$.

A common property to all the above examples is their regularity: All
the $\Vvertices$ have the same degree and similarly all the $\Zvertices$
have the same degree. An example we encountered of a non-regular graph
for which a mechanism  exists is $\EXAMPLEocExampleThreeBoxes$, 
 which can be represented as a non-regular \ZVlineGraphTEXT as 
$\ExampleFinner{c}{@R=1.5ex@C=1.5ex}{!}{.3em}$.

In the definition of the \ZVlineGraphsTEXT family we extend the above
family (and extend the mechanism accordingly) in two different ways:
allowing edges between the $\Zvertices$ (under a similar interval
constraint), and replacing vertices by a tree, a clique, or any other
\ZVlineGraphTEXT.  For example,
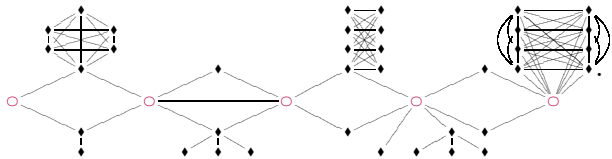
\vspace{-.7em}
\[
\ExampleBBuniformZVdecoratedinnerC{c}{@R=.6em}{.38\paperwidth}{!}{2}\text{.}
\]

In particular, the \ZVlineGraphsTEXT family includes all trees, cliques,
block graphs~\cite{Harary1963}, cycles of size up to $4$ (note
that there is no manipulation-resistant Pareto optimal anonymous mechanism
for cycles of size larger than~$5$), and all graphs for which (as
far as we found) a false-name-proof mechanism is known to the community.%

\subsection*{Related work}

Problems of facility location on discrete graphs were also studied
by Dokow et al.~\cite{Dokow2012}, who characterized the strategy-proof
mechanisms for the discrete line and discrete cycle. Other variants
of the facility location problem were also considered in the literature.
For instance, Schummer and Vohra~\cite{Schummer2002} considered
the case of continuous graphs, Lu et al.~\cite{lu2009tighter,LSWZ10}
studied variants in which several facilities need to be located and
scenarios in which an agent is located on several locations, and Feldman
et al.~\cite{DBLP:conf/sigecom/FeldmanFG16} studied the impact of
constraining the input language of the agents.

False-name-proofness was first introduced by Yokoo et al.~\cite[(based on a series of previous conference papers)]{DBLP:journals/geb/YokooSM04}
in the framework of combinatorial auctions. In this work, the authors
showed that the VCG mechanism does not satisfy false-name-proofness
in the general case, and they proposed a property of the preferences
under which this mechanism becomes false-name-proof. A similar concept
was also studied in the framework of peer-to-peer systems by Douceur~\cite{DBLP:conf/iptps/Douceur02}
under the name sybil attacks. Later, Conitzer and Sobel~\cite{DBLP:conf/wine/Conitzer08}
analyzed false-name-proof mechanisms in voting scenarios, Todo et
al.~\cite{DBLP:conf/atal/TodoIYS09} characterized other false-name-proof
mechanisms for combinatorial auctions, and Todo et al.~\cite{DBLP:conf/atal/TodoIY11}
characterized the false-name-proof mechanisms for facility location
on the continuous line and on continuous trees. In a recent work,
Ono et al.~\cite{DBLP:conf/prima/OnoTY17} showed, in the framework
of facility location on the discrete line, a relation between false-name-proofness
and the property of \emph{population monotonicity}.

The characterization of manipulation-resistant mechanisms for facility
location  is highly related to problems in \emph{Approximate mechanism
design without money}~\cite{DBLP:journals/teco/ProcacciaT13}. In
these problems, agents are characterized using cardinal utilities
and the designer seeks to find an outcome maximizing a desired target
function (e.g., sum of utilities, product of utilities, or minimal
utility). These works bound the trade-of between the target function
and manipulation-resistance, that is, they bound the loss to the target
function due to manipulation-resistance constraints. Similar bounds
were derived for  false-name-proof facility location mechanisms on
the continuous line and tree by  Todo et al.~\cite{DBLP:conf/atal/TodoIY11},
strategy-proof facility location on the continuous cycle by Alon
et al.~\cite{DBLP:journals/mor/AlonFPT10}, and for strategy-proof
facility location on the discrete line and cycle by Dokow et al.~\cite{Dokow2012}.

\subsubsection*{Approximate mechanism design}

In this work we do not analyze the approximation implications of
the characterization and in particular we do not assume a specific
cardinal representation of the agents' preferences.  Yet, we claim
 that for most natural representations and target functions the approximation
ratio is expected to be bad. For example, recall the above bi-clique
example. In this mechanism, the facility might be located on an `extremely'
left $\Zvertex$. Moreover,  the facility might be very far from
the vast majority of the agents, resulting in a very bad approximation
ratio for most reasonable target functions. This phenomenon is not
specific for the bi-clique graphs. For most \ZVlineGraphsTEXT, (due
to the false-name-proof requirement) the mechanism might be located
on a location extremely far from almost all agents, resulting in a
very bad approximation ratio (roughly, the number of agents times
the girth of the graph) for most reasonable target functions.%

\newpage{}

\section{Model}

Consider a graph $\GR=\left\langle \GRvNEW,\GRe\right\rangle $ with
a set of vertices $\GRvNEW$ and a set of, neither weighted nor directed,
edges $\GRe\subseteq\binom{\GRvNEW}{2}$, and we refer to the vertices
$v\in\GRvNEW$ also as \emph{locations} and use the two terms interchangeably.
 The distance between two vertices $v,u\in\GRvNEW$, notated $\DISTof vu$,
is the length of the shortest path connecting $v$ and $u$,\footnote{For simplicity, we assume the graph is connected.}
and the distance between a vertex $v\in\GRvNEW$ and set of vertices
$S\subseteq\GRvNEW$, $\DISTof vS$, is defined as the minimal distance
between $v$ and a vertex in $S$. We define $\BallOf vd$, the \emph{ball}
of radius $d\geqslant0$ around a vertex $v\in\GRvNEW$, to be the
set of vertices of distance at most $d$ from $v$,  $\BallOf vd=\left\{ u\in\GRvNEW\SetSt\DISTof vu\leqslant d\right\} $.
We say that two vertices are \emph{neighbors}  if there is an edge
connecting them  and notate by $\NeighOf v$ the set of neighbors
of a vertex $v$.

An instance of the \emph{facility location problem over $\GR$} 
is comprised of $n$ agents who are located on vertices of $\GRvNEW$;
Formally, we represent it by a \emph{location profile} $\LOCvec\in\GRvNEW^{n}$
where $\LOCof i$ is the location of Agent~$i$. Given an instance
$\LOCvec$, we would like to locate a facility on a vertex of the
graph while taking into account the preferences of the agents over
the locations. In this work, we assume the preference of an agent
is defined  by her distance to the facility:  An agent located on
$\LOCof{}\in\GRvNEW$ strictly prefers the facility being located
on $v\in\GRvNEW$ over it being located on $u\in\GRvNEW$ iff $\DISTof{\LOCof{}}v<\DISTof{\LOCof{}}u$.

A \emph{general facility location mechanism} (or shortly a \emph{mechanism})
defines for any profile of agents' locations a location for the facility.
We require the mechanism to assign a location for the facility for
any profile and any number of agents. Hence, we represent the mechanism
by a function $\MECH\colon\VtoTheStar\rightarrow\GRvNEW$. We also
think on $\MECH$ as a voting procedure: Each agent votes (and we
also refer to his vote as a \emph{ballot}) for a location, and based
on the ballots $\MECH$ returns a location for the facility. We say
that a mechanism is \textbf{anonymous} if the outcome $\MECHof{\LOCvec}$
does not depend on the identities of the agents, i.e., it can be defined
as a function of the ballot tally, the number of votes for each of
locations.

\subsection*{Manipulation-resistance}

A strategic agent might act untruthfully if she thinks it might cause
the mechanism to return a location she prefers (that is, a location
closer to her). In this work we consider the following manipulations:
\textbf{ Misreport}: An agent  might report to the mechanism a location
different from her real location; \textbf{ False-name-report}: An
agent might pretend to be several agents and submit several (not necessarily
identical) ballots; \textbf{ Abstention}: An agent might choose
not to participate in the mechanism at all. A mechanism in which no
agent benefits from these manipulations, regardless to the ballots
of the other agents, is said to be \textbf{strategy-proof}, \textbf{false-name-proof},
and \textbf{abstention-proof}, respectively. We also consider a generalization
of these manipulations  to manipulations of a coalition, and say a
mechanism is \textbf{group-manipulation-resistant} (shortly manipulation-resistant)
if no coalition can change the outcome, by misreporting, false-name-reporting,
or abstaining, to a different location which they unanimously agree
is no worse than the original outcome (i.e., if they vote sincerely)
and at least one of the coalition's members strictly prefers the new
location. 
\begin{defn}[Group-manipulation-resistant]
\MyLyxThmFN{For simplicity of notations, we give the formal definition for anonymous
mechanisms.} A mechanism $\MECH$ is  group\hyp{}manipulation\hyp{}resistant
if  there exists no coalition of agents $C\subseteq\left\{ 1,\ldots,n\right\} $,
a vector of locations $\LOCvec\in\GRvNEW^{n}$, and a set of ballots
$\REPORTset\in\VtoTheStar$ s.t. $\boldsymbol{\left(\rmnum 1\right)}$
all the members of $C$ weakly prefer $\MECHof{\REPORTset,\LOCvec_{-C}}$,
that is, the outcome when the agents outside of $C$ do not change
their vote and the agents of $C$ replace their ballots by $\REPORTset$,
over $\MECHof{\LOCvec}$ and $\boldsymbol{\left(\rmnum 2\right)}$
at least one of $C$'s members strictly prefers $\MECHof{\REPORTset,\LOCvec_{-C}}$
over $\MECHof{\LOCvec}$.
\end{defn}

We note that for $C=\left\{ i\right\} $ being a singleton, this general
manipulation coincides with misreport for $\sizeof{\REPORTset}=1$,
with false-name-report for $\sizeof{\REPORTset}\!>\!1$, and with
abstention for $\REPORTset\!=\!\emptyset$.

\subsubsection*{The revelation principle}

One could consider more general mechanisms in which the agents vote
using more abstract ballots, and define similar manipulation-resistance
terms for the general framework. Applying a simple direct revelation
principle~\cite{Myerson79} shows that any such general manipulation-resistant
mechanism is equivalent to a manipulation-resistant mechanism in our
framework: The two mechanisms implement the same mapping of the agents
private preferences to a location for the facility, and since the
above properties are defined for the mapping they are invariant to
this transformation.  That is, given some general mechanism $M$
that maps abstract actions to a location for the facility and a behavior
protocol $D$ that maps types of the agents (i.e., locations) to actions
of $M$, if $D$ satisfies the generalized desiderata, then the direct
mechanism $M\circ D$ satisfies our desiderata (w.r.t. truth-telling).

\subsection*{Efficiency}

So far, we defined the desired manipulation-resistance properties
for a mechanism. On the other hand, we would also like the mechanism
to respect the preferences of the agents. We would like to avoid a
scenario in which, after the mechanism has been used, the agents can
agree that a different location is preferable. Given a location profile
$\LOCvec\in\GRvNEW^{n}$, the set of \emph{Pareto optimal} locations,
$\POof{\LOCvec}$, is the set of all locations which the agents cannot
agree to rule out.  Formally, given two locations $v,u\in\GRvNEW$,
we say that $u$ \emph{Pareto dominates} $v$ (w.r.t. a location profile
$\LOCvec$) if $\boldsymbol{\left(\rmnum 1\right)}$ all agents weakly
prefer $u$ over $v$ and $\boldsymbol{\left(\rmnum 2\right)}$ at
least one agent strictly prefers $u$ over $v$. We say that $v$
is \emph{Pareto optimal} ($v\in\POof{\LOCvec}$) if it is not Pareto
dominated by any other location. We say a mechanism is \textbf{Pareto
optimal} if for any report profile $\LOCvec$ (and assuming truthful
reporting) $\MECHof{\LOCvec}\in\POof{\LOCvec}$. In particular, Pareto
optimality entails \textbf{unanimity}, if all the agents unanimously
vote for the same location then the mechanism outputs this location,
and \textbf{citizen sovereignty}, the mechanism is onto and does not
a-priori rule out any location.%

\newpage{}

\section{\label{sec:Main-Result}Main Result}

In this work, we define a family of graphs, \emph{\ZVlineGraphsTEXT
}, and present a general mechanism for this family. This family is
defined by introducing a simple combinatorial structure - partition
to two types of vertices and connectivity constraint. One could think
of the partition as representing a social agreement according to which
the mechanism is defined, e.g., a subset of status-quo locations or
an a-priori priority hierarchy over the locations. The connectivity
constraint (as the graph in general) represents the homogeneity over
the agents' preferences which allows us to show a manipulation-resistant
mechanism.
\begin{defn}[\ZVopartTEXT]
\label{def:ZV-partition}  Given an unweighted undirected connected graph $\GR=\left(\GRvNEW,\GRe\right)$
and a sequence of non-empty sets of vertices $\GRzVerts,\SubGraph 1,\ldots\allowbreak,\SubGraph k\allowbreak\subseteq\GRvNEW$,
we say that the sequence $\GRzVerts,\SubGraph 1,\ldots,\SubGraph k$
($k\geqslant0$) is a \ZVopartTEXT if the following holds.
\begin{enumerate}
\item The sets $\SubGraph i$ are disjoint, $\SubGraph i\cap\SubGraph j=\emptyset\text{\ensuremath{\quad}for }i\neq j$.
\item The sequence is a cover of $\GRvNEW$, $\GRzVerts\cup\SubGraph 1\cup\cdots\cup\SubGraph k=\GRvNEW$,
and no sub-sequence of it is a cover of $\GRvNEW$. 
\item For $i=1,\ldots,k$, there is a unique vertex in $\SubGraph i$ which
is closest to $\GRzVerts$. We refer to it as \emph{the root of $\SubGraph i$}
and denote it by $\RootOfSubGraph i$,\vspace{-.3em}
\[
\RootOfSubGraph i=\argmin_{v\in\SubGraph i}\DISTof v{\GRzVerts}\text{.}
\]
\vspace{-1em}
\item \label{enu:No intra-edges}All paths between vertices of $\SubGraph i$
and vertices outside of $\SubGraph i$ pass through the root $\RootOfSubGraph i$
and through~$\GRzVerts$.
\item Last, $\GRzVerts$ is equipped with an order (that is, an injective
mapping from $\GRzVerts$ to $\Re$). For simplicity of description,
we refer to this order as an order from left to right. We call a subset
$A$ of $\GRzVerts$ an \emph{interval} if $A$ is the preimage of
an interval in $\Re$, , i.e., if it is a sequence of vertices according
to the order.
\end{enumerate}
We use the notions $\ViSubGraphs$, $\Vvertices$, and $\Zvertices$
for the respective sets of vertices. Note that we do not require the
sets of the $\Zvertices$ and the $\Vvertices$ to be disjoint. For
instance, in the last example of the introduction the $9$-clique
includes the rightmost $\Zvertex$. Notice that from the third condition
it is clear that for all $i$ the intersection $\SubGraph i\cap\GRzVerts$
is of size at most one. 
\end{defn}

Given a graph $\GR=\left(\GRvNEW,\GRe\right)$ with a \ZVopartTEXT,
$\GRzVerts,\SubGraph 1,\allowbreak\ldots\allowbreak,\SubGraph k\subseteq\GRvNEW$,
and a sequence of mechanisms $\MECH_{i}\colon\XtoTheStar{\left(\SubGraph i\right)}\rightarrow\SubGraph i$
for $i=1,\ldots,k$, we define the following mechanism $\TheMECH\colon\VtoTheStar\rightarrow\GRvNEW$: 
\begin{defn}[$\TheMECH$]
\label{def:ZV-mechanism}Given a vector of reports $\LOCvec\in\VtoTheStar$ 
\begin{itemize}
\item[$\blacktriangleright$]  If all the ballots belong to the same $\SubGraph i$-subgraph, return
$\MECH_{i}\left(\LOCvec\right)$.
\item[$\blacktriangleright$]  Otherwise, return the leftmost Pareto optimal location in $\GRzVerts$.
\end{itemize}
\end{defn}

It is not hard to see the following: $\bullet$ $\TheMECH$ is well defined (If it does not hold that all
ballots belong to the same $\SubGraph i$-subgraph, then necessarily
$\POof{\LOCvec}\cap\GRzVerts\neq\emptyset$), $\bullet$ $\TheMECH$
runs in polynomial time, and $\bullet$ If $\MECH_{1},\ldots,\MECH_{k}$
can be defined as the first Pareto optimal location according to some
order, then an equivalent way to define $\TheMECH$ is as the first
Pareto optimal location in the following order:  First, go over the
vertices of $\GRzVerts$ from left to right, and then on the vertices
of the $\SubGraph i$-subgraphs in some order s.t. for each subgraph
the order over its vertices matches the order of $\MECH_{i}$.

Next, we define \ZVlineGraphsTEXT by introducing a connectivity constraint.
\begin{defn}[\ZVlineGraphTEXT]
\label{def:ZV-line} 

An unweighted undirected connected graph $\GR=\left(\GRvNEW,\GRe\right)$
is a \ZVlineGraphTEXT w.r.t. $\WRTaPartition$, if   $\boldsymbol{\left(a\right)}$
$\left\langle \GRzVerts,\SubGraph 1,\ldots\allowbreak,\SubGraph k\right\rangle $
is a \ZVopartTEXT of $\GR$,  $\boldsymbol{\left(b\right)}$ for
any vertex $z\in\GRzVerts$, $\BallOfCapZ z1$ is an interval in $\GRzVerts$,
 and if $k>0$ then for $i=1,\ldots,k$  $\boldsymbol{\left(c\right)}$
the induced graph $\GR_{i}=\left\langle \SubGraph i,\GRe\cap\left(\SubGraph i\times\SubGraph i\right)\right\rangle $
is a \ZVlineGraphTEXT,  $\boldsymbol{\left(d\right)}$  $\RootOfSubGraph i$
is a $\Zvertex$ of $\GR_{i}$ (that is a $\Zvertex$ in the representation
of $\GR_{i}$), and it is the leftmost $\Zvertex$ of f $\GR_{i}$,
and last $\boldsymbol{\left(e\right)}$ $\BallOfCapZ{\RootOfSubGraph i}1$
is an interval in $\GRzVerts$.
\end{defn}

For instance, for any $\ell\geqslant1$ the clique over $\ell$ vertices,
$\Clique{\ell}$, is a \ZVlineGraphTEXT w.r.t. $\GRzVerts=\GRvNEW$
and any order over the vertices. 

Given a \ZVlineGraphTEXT $\GR=\left\langle \GRvNEW,\GRe\right\rangle $,
applying Def.~\ref{def:ZV-mechanism} recursively on $\GR$ and its
$\SubGraph i$-subgraphs gives us a mechanism $\TheMECH\colon\VtoTheStar\rightarrow\GRvNEW$.
Note that $\TheMECH$ depends on the representation of $\GR$ w.r.t.
a specific \ZVopartTEXT, and in case that $\GR$ can be represented
as a \ZVlineGraphTEXT w.r.t. several \ZVopartsTEXT, they might result
in different mechanisms. Our main result shows that this mechanism
satisfies the desired properties. 

\begin{thm}[Main result]
\label{thm:MainThm-Ver2}
Let $\GR=\left(\GRvNEW,\GRe\right)$ be a \ZVlineGraphTEXT w.r.t.
$\WRTaPartition$ and let $\TheMECH\colon\VtoTheStar\rightarrow\GRvNEW$
be the result of applying Definition~\ref{def:ZV-mechanism}. recursively
on $\GR$. Then $\TheMECH$ is an anonymous Pareto optimal mechanism
and $\TheMECH$ satisfies:

\hfill\begin{minipage}[t]{.95\linewidth}

For any vector of locations $\LOCvec\in\left(\SubGraph i\right)^{n}$,
 a coalition of agents $C$,  and a set of ballots $\REPORTset\in\XtoTheStar{\left(\SubGraph i\right)}$,\footnotemark{}
 $\REPORTset$ is not a beneficial deviation for $C$ (That is, $C$
does not strictly prefer $\TheMECHof{\REPORTset,\LOCof{-C}}$ over
$\TheMECHof{\LOCvec}$.\footnotemark{}

\end{minipage}\addtocounter{footnote}{-2}\stepcounter{footnote}\footnotetext{Since
$\TheMECH$ is an anonymous mechanisms, we define $\REPORTset$ as
a set of ballots ignoring identities.}\stepcounter{footnote}\footnotetext{Since
$\TheMECH$ is onto, this property entails Pareto optimality. Yet,
we prefer to state explicitly Pareto optimality as a desired efficiency
property.}
\end{thm}

{\newcommand{\TEXTwidth}{0.785}%
\edef\RememberIndent{\the\parindent}%
\noindent\begin{minipage}[t]{\TEXTwidth\linewidth}\setlength{\parindent}{\RememberIndent}%

Note  that the theorem does not hold for weighted graphs (that is,
when edges have non-uniform length). Consider the following weighted
graph and a profile in which Alice is located on $z_{r}$ and Bob
on $v$. Then, the outcome is $z_{r}$, but Bob can move the facility
to a preferred location $z_{\ell}$  both $\boldsymbol{\left(\rmnum 1\right)}$
by misreporting $z_{\ell}$, hence $\TheMECH$ is not strategy-proof,
and $\boldsymbol{\left(\rmnum 2\right)}$ by false-name-reporting
$z_{\ell}$ in addition to his sincere report, hence $\TheMECH$ is
not false-name-proof.

\end{minipage}\hfill\begin{minipage}[t][1\totalheight][c]{\dimexpr1\linewidth -\TEXTwidth\linewidth -0.02\linewidth}%

$\ExampleCounterExampleWeightedGraph[@R=0.8em]$

\end{minipage}} 

We note that there are trivial mechanisms which satisfy  subsets of
these properties: %
\textbullet{} The \emph{fixed mechanism}, which locates the facility on a pre-defined
location ignoring the votes of the agents, is trivially manipulation-resistant
and anonymous, but it is not onto and hence not Pareto optimal.  %
\textbullet{} A \emph{dictatorship}, e.g., the mechanism that always locates the
facility on the location reported by the first agent, is not anonymous
but clearly it is manipulation-resistant.\footnote{While we did not formally define false-name-proofness for non-anonymous
mechanisms, assuming a false-name vote cannot be counted as the vote
of the first agent, no agent can benefit from casting additional ballots.} %
\textbullet{}  The \emph{median mechanism}, which minimizes the sum of distances
between the facility and the ballots, is anonymous and Pareto optimal,
and it is not hard to see that for the discrete line it satisfies
strategy-proofness and abstention-proofness both against one manipulator
and against a coalition but an agent will benefit by casting multiple
identical ballots. %
\textbullet{} The \emph{mean mechanism}, which minimizes the sum of squares of the
distances between the facility and the ballots, is anonymous and Pareto
optimal but might not be strategy-proof or false-name-proof even against
one agent, e.g., for the discrete line graph (it is abstention-proof,
though).

\subsection{Implications: Mechanisms for recursive graph families\label{subsec:Implications:-Mechanisms-for}}

By applying the main result to a recursive graph family, we can generate
a recursive (and hence commonly simple) mechanism which satisfies
our desiderata. For instance, a corollary of our result is a manipulation-resistant
mechanism for the following  family of rooted graphs (that is, $\left\langle \GRvNEW,\GRe,r\right\rangle $
s.t. $\GRe\subseteq\binom{\GRvNEW}{2}$ and $r\in\GRvNEW$).
\begin{defn}[$\mathcal{F}$]
\MyLyxThmNewline{}
\begin{itemize}
\item $\left\langle \left\{ v\right\} ,\emptyset,v\right\rangle \in\mathcal{F}$.
\item For any $k,\ell\geqslant1$: If $\left\{ \left\langle \GRvNEW_{i},\GRe_{i},r_{i}\right\rangle \right\} _{i=1}^{k}$
are in $\mathcal{F}$ (and the $\GRvNEW_{i}$ are disjoint), then
also the following graph is in $\mathcal{F}$.\vspace{-.7em}
\[
\left(\left\{ \widehat{r}_{j}\right\} _{j=1}^{\ell}\DisjCup\left(\bigDisjCup_{i=1}^{k}\GRvNEW_{i}\right),\left\{ \left\langle \widehat{r}_{j},r_{i}\right\rangle \right\} _{i=1\ldots k,j=1\ldots\ell}\DisjCup\left(\bigDisjCup_{i=1}^{k}E_{i}\right),\widehat{r}_{1}\right)
\]
I.e., adding a new layer of pre-roots, a bi-clique between them and
the roots of the graphs of the previous stage, and defining one of
the pre-roots to be the new root.
\end{itemize}
\end{defn}

\begin{claim}
The anonymous Pareto optimal mechanism $\MECHof{\LOCvec}=\argmin_{v\in\PO\left(\LOCvec\right)}\allowbreak\DISTof vr$,
which returns the Pareto optimal location closest to the root and
breaks ties according to a predefined order, is manipulation-resistant.
\end{claim}

Note that by setting $\ell=1$ in the second step of the definition
we get a recursive definition of \emph{rooted trees}. Hence, we get
that for any tree $\GR$ the mechanism that returns the lowest common
ancestor of the ballots (with regard to some root) is a manipulation-resistant
mechanism (These are also the mechanisms which Todo et al.~\cite{DBLP:conf/atal/TodoIY11}
characterized as the false-name-proof, anonymous, and Pareto optimal
mechanisms for the continuous tree.). 
\begin{proof}
We prove the claim by induction over, $h\left(\GR\right)$, the number
of steps needed to generate $\GR$.

\uline{If \mbox{$h\left(\GR\right)=0$}}, i.e., $\GR=\left\langle \left\{ v\right\} ,\emptyset,v\right\rangle $
consists of a single vertex and the trivial mechanism satisfies all
the desired properties.

\uline{If \mbox{$h\left(\GR\right)\geqslant1$}}, then $\GR$
is a \ZVlineGraphTEXT w.r.t. $\GRzVerts=\left\{ \widehat{r}_{j}\right\} _{j=1}^{\ell}$
and $\SubGraph i=\GRvNEW_{i}$. Note that for all $\SubGraph i$-subgraphs
$h\left(\left\langle \GRvNEW_{i},\GRe_{i},r_{i}\right\rangle \right)\leqslant h\left(\GR\right)-1$.
Hence, our recursive mechanism returns one of the pre-roots of the
`lowest' subgraph which includes $\LOCvec$ when ties are broken according
to the (arbitrary) order over the pre-roots.
\end{proof}
A second example is \emph{connected block graphs~\cite{Harary1963}}.
A connected graph $\GR=\left\langle \GRvNEW,\GRe\right\rangle $ is
a block graph if the following equivalent conditions hold:
\begin{itemize}
\item Every biconnected component of $\GR$ is a clique. (Since for any
graph the structure of its biconnected components  is described by
a block-cut tree, connected block graphs are also called \emph{clique
trees}.)
\item The intersection of any two connected subgraphs of $\GR$ is either
empty or connected.
\item For every four vertices $u,v,w,x\in\GRvNEW$, the larger two of the
distance sums $d\left(u,v\right)+d\left(w,x\right)$, $d\left(u,w\right)+d\left(v,x\right)$,
and $d\left(u,x\right)+d\left(v,w\right)$  are equal.
\end{itemize}
Our mechanism for a connected block graph $\GR$ returns the closest
Pareto optimal location to an arbitrarily predefined location, breaking
ties according to an arbitrarily predefined order over the locations.

\begin{proof-sketch}

$\GR$ is connected block graph and hence all bi-connected components
of $\GR$ are cliques. The \emph{block-cut tree} of $\GR$ is a tree
$\mathcal{T}\left(\GR\right)$ which is defined in the following way.
In $\mathcal{T}\left(\GR\right)$ there is a vertex (\emph{component-vertex})
for each maximal biconnected component of $\GR$ and a vertex (\emph{intersection-vertex})
for each vertex in $\GR$ which belongs to more than one maximal biconnected
component. There is an edge in $\mathcal{T}\left(\GR\right)$ between
each component-vertex and the intersection-vertices belonging to this
component.

Following the inductive structure of $\mathcal{T}\left(\GR\right)$,
and recalling that a clique is a \ZVlineGraphTEXT w.r.t. all vertices
of the clique being $\Zvertices$ and any order over them, we get
that our mechanism is defined by an arbitrary predefined component-vertex
of $\mathcal{T}\left(\GR\right)$, $\RootSymbol$, and a series of
arbitrary predefined orders over the locations of each of the components.
The mechanism is:
\begin{itemize}
\item[$\blacktriangleright$] If all ballots belong to the same component, return the first location
(according to the order) that was voted for.
\item[$\blacktriangleright$] Otherwise, choose the component closest to $\RootSymbol$ s.t. one
of the locations of the component is Pareto optimal, and return the
first location (according to the order) in this component.
\end{itemize}
Last, we note that an equivalent definition of this mechanism is returning
the closest Pareto optimal location to some location $v\in\RootSymbol$,
breaking ties according to a concatenation of the orders over the
components.
\end{proof-sketch}%

\newpage{}

\section{Proof of Main Result (Thm.~\ref{thm:MainThm-Ver2})\label{sec:Proof-of-MainThm}}

We prove a stronger result which shows a general method for generating
a mechanism~$\TheMECH$ (satisfying the desired properties) for a
given graph from mechanisms for its subgraphs,~$\MECH_{i}$. Theorem~\ref{thm:MainThm-Ver2}
is an immediate special case of this lemma. The same proof shows that
also for weaker manipulation-resistance properties, e.g., against
individual agents, against misreporting, or against abstentions, manipulation-resistance
of the mechanisms for the subgraphs~$\MECH_{i}$, result in the same
manipulation-resistance notion for the mechanism of the graph~$\TheMECH$.
\begin{lem}
\label{thm:MainLemmaForMainThm}Let $\GR=\left(\GRvNEW,\GRe\right)$ be a graph with a \ZVopartTEXT
$\WRTaPartition$ and let $\MECH_{i}\colon\XtoTheStar{\left(\SubGraph i\right)}\rightarrow\SubGraph i$
be a sequence of mechanisms s.t. for $i=1,\ldots,k$
\begin{itemize}
\item $\MECHi i$ is anonymous and Pareto optimal;
\item For an infinite number of $\tau\in\N$, there exists a profile $\LOCvec\in\XtoTheStar{\left(\SubGraph i\right)}$
in which all locations in $\SubGraph i$ were voted for at least $\tau$
times and $\MECHiOf i{\LOCvec}=\RootOfSubGraph i$; and
\end{itemize}
\begin{itemize}
\item {\newcommand{\tikzmark}[1]{%
    \tikz[baseline={(#1.base)},remember picture]%
    \node[outer sep=0pt, inner sep=0pt] (#1) {\strut};}%
\newcommand{\TEXTwidth}{0.935}%
\begin{minipage}[t]{\TEXTwidth\linewidth}

For any vector of locations $\LOCvec\in\left(\SubGraph i\right)^{n}$,
 a coalition of agents $C$,  and a set of ballots $\REPORTset\in\XtoTheStar{\left(\SubGraph i\right)}$,\footnotemark{}
 $\REPORTset$ is not a beneficial deviation for $C$ (That is, $C$
does not strictly prefer $\MECHiOf i{\REPORTset,\LOCof{-C}}$ over
$\MECHiOf i{\LOCvec}$).

\end{minipage}\begin{minipage}[t]{\dimexpr1\linewidth -\TEXTwidth\linewidth -0.05\linewidth}%
\tikzmark{TOP}\\%
{}\\%
\tikzmark{BOT}%
\end{minipage}%
\tikz[remember picture,overlay,thick,
	decoration={calligraphic brace,
    	raise=2pt, amplitude=6pt},%]%
	pen colour={cyan!80!black}]%
\draw[decorate] (TOP.north) -- (TOP |- BOT.south) node [black,midway,xshift=18pt]{$\ThePropSymbol$};\footnotetext{Since $\MECH_{i}$ (and later $\TheMECH$)
are anonymous mechanisms, we define $\REPORTset$ as a set of ballots
ignoring identities.}}
\end{itemize}

Then, for $\TheMECH\colon\VtoTheStar\rightarrow\GRvNEW$ as defined
in Definition~\ref{def:ZV-mechanism}, $\TheMECH$ is an anonymous
and Pareto optimal mechanism and
\begin{itemize}
\item[$\left(\boldsymbol{\Rmnum 1}\right)$] If $\GR$ is a \ZVlineGraphTEXT w.r.t. $\WRTaPartition$, then $\TheMECH$
satisfies $\ThePropSymbol$.
\item[$\left(\boldsymbol{\Rmnum 2}\right)$] If  $\RootOfSubGraph i\in\GRzVerts$ for all $i=1,\ldots,k$, and
the mechanism $\MECH_{\GRzVerts}\colon\XtoTheStar{\GRzVerts}\allowbreak\rightarrow\allowbreak\GRzVerts$
which returns the leftmost Pareto optimal location satisfies $\ThePropSymbol$,
then also $\TheMECH$ satisfies~$\ThePropSymbol$.
\end{itemize}
\end{lem}

\begin{proof}

The anonymity of $\TheMECH$ is an immediate corollary of the mechanisms
$\MECHi i$ and $\MECH_{\GRzVerts}$ being anonymous mechanisms.

Notice that if all agents are in the same $\SubGraph i$-subgraph,
then all of them strictly prefer $\RootOfSubGraph i$ over any location
outside of $\SubGraph i$, so $\POof{\LOCvec}\subseteq\SubGraph i$.
Moreover, any location $v\in\SubGraph i\setminus\POof{\LOCvec}$ is
Pareto dominated by a location $y\in\POof{\LOCvec}\subseteq\SubGraph i$.
Hence, the Pareto optimal set when considering only the locations
in $\SubGraph i$ equals to the Pareto optimal set when considering
all locations. Since, the mechanisms $\MECHi i$ are Pareto optimal
mechanisms we get that also $\TheMECH$ is Pareto optimal.

In order to prove the main part of the theorem, we assume towards
a contradiction that there exists a vector of locations $\LOCvec\in\GRvNEW^{n}$,
a coalition of agents $C$, and a set of ballots $\REPORTset\in\VtoTheStar$,
s.t. $C$ can, by voting $\REPORTset$, get an outcome $\TheMECHof{\REPORTset,\LOCvec_{-C}}$
which it strictly prefers, that is, all of its members weakly prefer
$\TheMECHof{\REPORTset,\LOCvec_{-C}}$ over $\TheMECHof{\LOCvec}=\TheMECHof{\LOCvec_{C},\LOCvec_{-C}}$,
and at least one of $C$'s members, Agent~$i$ for $i\in C$, strictly
prefers $\TheMECHof{\REPORTset,\LOCvec_{-C}}$ over $\TheMECHof{\LOCvec}$.
 $\TheMECHof{\LOCvec}\in\POof{\LOCvec}$ and in particular the coalition
of all agents does not strictly prefer $\TheMECHof{\REPORTset,\LOCvec_{-C}}$
over $\TheMECHof{\LOCvec}$. Hence, there exists an Agent~$j$, for
$j\notin C$, who strictly prefers $\TheMECHof{\LOCvec}$ over $\TheMECHof{\REPORTset,\LOCvec_{-C}}$.

\uline{If \mbox{$\TheMECHof{\LOCvec}$} is not in \mbox{$\GRzVerts$}}:
Then necessarily, all the locations in $\LOCvec$ and $\TheMECHof{\LOCvec}$
belong to the same $\SubGraph i$-subgraph, w.l.o.g. $\SubGraph 1$,
so $\TheMECHof{\LOCvec}=\MECHiOf 1{\LOCvec}$. Since $\MECHi 1$ is
resistant to false-name manipulations of Agent~$i$ and since Agent~$i$
can achieve $\RootOfSubGraph 1$ by casting enough false ballots,
we get that Agent~$i$ weakly prefers $\TheMECHof{\LOCvec}$ over
$\RootOfSubGraph 1$ and hence Agent~$i$ strictly prefers $\TheMECHof{\REPORTset,\LOCvec_{-C}}$
over $\RootOfSubGraph 1$. Since for any $u$ outside of $\SubGraph 1$
it holds that $\DISTof{\LOCof i}{\RootOfSubGraph 1}<\DISTof{\LOCof i}u$,
we get that $\TheMECHof{\REPORTset,\LOCvec_{-C}}\in\SubGraph 1\setminus\RootOfSubGraph 1\subseteq\SubGraph 1\setminus\GRzVerts$.
Hence, $\REPORTset\subseteq\SubGraph 1$ and $\TheMECHof{\REPORTset,\LOCvec_{-C}}=\MECHiOf 1{\REPORTset,\LOCvec_{-C}}$,
and we get a contradiction to the false-name-proofness of $\MECHi 1$.

Similarly, \uline{if \mbox{$\TheMECHof{\REPORTset,\LOCvec_{-C}}$}
is not in \mbox{$\GRzVerts$}}: Then necessarily, $\TheMECHof{\REPORTset,\LOCvec_{-C}}$
and all the locations in $\REPORTset$ and $\LOCvec_{-C}$ belong
to the same $\SubGraph i$-subgraph, w.l.o.g. $\SubGraph 1$, so $\TheMECHof{\REPORTset,\LOCvec_{-C}}=\MECHiOf 1{\REPORTset,\LOCvec_{-C}}$.
Since $\MECHi 1$ is resistant to false-name manipulations of Agent~$j$
and since Agent~$j$ can achieve $\RootOfSubGraph 1$ by casting
enough false ballots, we get that Agent~$j$ weakly prefers $\TheMECHof{\REPORTset,\LOCvec_{-C}}$
over $\RootOfSubGraph 1$ and strictly prefers $\TheMECHof{\LOCvec}$
over $\RootOfSubGraph 1$. Since for any $u$ outside of $\SubGraph 1$
it holds that $\DISTof{\LOCof j}{\RootOfSubGraph 1}<\DISTof{\LOCof j}u$,
we get that $\TheMECHof{\LOCvec}\in\SubGraph 1\setminus\RootOfSubGraph 1\subseteq\SubGraph 1\setminus\GRzVerts$.
Hence, $\LOCvec\subseteq\SubGraph 1$ and $\TheMECHof{\LOCvec}=\MECHiOf 1{\LOCvec}$,
and we get a contradiction to the false-name-proofness of $\MECHi 1$.

\uline{If both \mbox{$\TheMECHof{\LOCvec}$} and \mbox{$\TheMECHof{\REPORTset,\LOCvec_{-C}}$}
are in \mbox{$\GRzVerts$}}: We deal with this case using two different
argumentations for the two scenarios of the theorem.
\begin{adjustwidth}{.04\linewidth}{}

\noindent{}\uline{\mbox{$\left(\boldsymbol{\Rmnum 1}\right)$}
\mbox{$\GR$} is a \ZVlineGraphTEXT w.r.t. \mbox{$\WRTaPartition$}}:
We first prove the following two auxiliary lemmas.
\begin{innerLem}
\label{lem:BallIsIntervalOnZ} For any $v\in\GRvNEW$ and $d\geqslant0$,
$\BallOfCapZ vd$ is an interval in $\GRzVerts$.
\end{innerLem}

\begin{proof}
We prove the lemma by induction over~$d$.

\uline{For \mbox{$d=0$}}, $\BallOfCapZ v0$ equals to $\left\{ v\right\} $
if $v\in\GRzVerts$ and to the empty set if $v\notin\GRzVerts$. 

\uline{For \mbox{$d=1$}}, $\BallOfCapZ v1$ is either the empty
set or an interval in $\GRzVerts$.

\uline{For \mbox{$d\geqslant2$}}: If $d<\DISTof v{\GRzVerts}$,
$\BallOfCapZ vd=\emptyset$. If $d\geqslant\DISTof v{\GRzVerts}>1$
(in particular, $v\notin\GRzVerts$ and is not a root\footnote{An easy corollary of the definition of \ZVopartTEXT is that for all
$\SubGraph i$-subgraphs $\DISTof{\RootOfSubGraph i}{\GRzVerts}\leqslant1$.}), then there exists a location $u$ (the root of the $\ViSubGraph$
$v$ belongs to) s.t. all paths from $v$ to locations in $\GRzVerts$
pass through $u$, $1\leqslant\DISTof vu\leqslant\DISTof v{\GRzVerts}\leqslant d$
and 
\[
\BallOfCapZ vd=\BallOfCapZ u{d-\DISTof vu}
\]
 which is an interval by the induction hypothesis. 

Otherwise, $\DISTof v{\GRzVerts}\leqslant1<d$ and in particular $\BallOfCapZ vd\neq\emptyset$,
and hence
\[
\BallOfCapZ vd=\left(\BallOfCapZ v1\right)\cup\left(\bigcup_{\substack{u\in\NeighOf v\st\\
\DISTof u{\GRzVerts}\leqslant1
}
}\BallOfCapZ u{d-1}\right)\text{.}
\]

For any $u\in\NeighOf v$ s.t. $\DISTof u{\GRzVerts}\leqslant1$ we
claim that $\BallOfCapZ u{d-1}$ and $\BallOfCapZ v1$ intersect.
\begin{itemize}
\item If $u\in\GRzVerts$: $u\in\left(\BallOfCapZ u{d-1}\right)\cap\left(\BallOfCapZ v1\right)$.
\item If $u\notin\GRzVerts$: then $v\in Z$ and $v\in\left(\BallOfCapZ u{d-1}\right)\cap\left(\BallOfCapZ v1\right)$. 
\end{itemize}
Hence, for any $u\in\NeighOf v$ s.t. $\DISTof u{\GRzVerts}\leqslant1$,
$\BallOfCapZ u{d-1}$ and $\BallOfCapZ v1$ are intersecting intervals
in $\GRzVerts$. So $\BallOfCapZ vd$ is an interval as the union
of intersecting intervals. 
\end{proof}
\begin{innerLem}
\label{lem:LeftMostPO} Let $\LOCvec$ be a vector of locations s.t.
$\TheMECHof{\LOCvec}\in\GRzVerts$ and let $v\in\GRzVerts$ be a location
s.t. Agent~$i$ strictly prefers $v$ over $\TheMECHof{\LOCvec}$.
Then $\TheMECHof{\LOCvec}$ is to the left of $v$. 
\end{innerLem}

\begin{proof}
If $\LOCof i\in\GRzVerts$ then $\LOCof i\in\POof{\LOCvec}\cap\GRzVerts$
and by the definition of $\TheMECH$, $\TheMECHof{\LOCvec}$ is to
the left of $\LOCof i$. Since $\TheMECHof{\LOCvec}\notin\BallOfCapZ{\LOCof i}{\DISTof{\LOCof i}v}$
and since this set is an interval which includes $\LOCof i$, we get
that $\TheMECHof{\LOCvec}$ is to the left of the interval and in
particular to the left of $v$.

Otherwise, $\LOCof i\notin\GRzVerts$ and there exists an Agent~$k$
for which $\LOCof k$ is not in the same $\SubGraph i$-subgraph as
$\LOCof i$. Hence, there exists a location $u\in\GRzVerts$   s.t.
$u$ is on a shortest-path from $\LOCof i$ to $\LOCof k$, $u\in\GRzVerts$,
and $u\in\POof{\LOCvec}$. Hence, $\DISTof{\LOCof i}u\leqslant\DISTof{\LOCof i}v$
and so
\[
u\in\BallOfCapZ{\LOCof i}{\DISTof{\LOCof i}u}\subseteq\BallOfCapZ{\LOCof i}{\DISTof{\LOCof i}v}\text{.}
\]
The two sets are intervals in $\GRzVerts$, $\TheMECHof{\LOCvec}$
is to the left of $u$ (or equal to it), and $\TheMECHof{\LOCvec}\notin\BallOfCapZ{\LOCof i}{\DISTof{\LOCof i}v}\text{.}$Hence,
$\TheMECHof{\LOCvec}$ is to the left of $v$.
\end{proof}
By applying Lemma~\ref{lem:LeftMostPO} for the profile $\LOCvec$
and Agent~$i$, we get that $\TheMECHof{\LOCvec}$ is to the left
of $\TheMECHof{\REPORTset,\LOCvec_{-C}}$; and by applying Lemma~\ref{lem:LeftMostPO}
for the profile $\left(\REPORTset,\LOCvec_{-C}\right)$ and Agent~$j$,
we get that $\TheMECHof{\REPORTset,\LOCvec_{-C}}$ is to the left
of $\TheMECHof{\LOCvec}$. Hence, we get a contradiction.

~

\noindent{}\uline{\mbox{$\left(\boldsymbol{\Rmnum 2}\right)$}
\mbox{$\RootOfSubGraph i\in\GRzVerts$} and \mbox{$\MECH_{\GRzVerts}$}
satisfies \mbox{$\ThePropSymbol$}}: We notice that since $\RootOfSubGraph i\in\GRzVerts$
for all $\SubGraph i$-subgraphs the preference of an agent which
is located in a $\SubGraph i$-subgraph over the locations in $\GRzVerts$
and an agent which is located on the root, $\RootOfSubGraph i$, are
identical. Hence, for any profile $\LOCvecGen y$ if $\TheMECHof{\LOCvecGen y}\in\GRzVerts$
then $\TheMECHof{\LOCvecGen y}=\MECH_{\GRzVerts}\left(\widehat{\LOCvecGen y}\right)$
for $\widehat{\LOCvecGen y}$ being the profile generated from $\LOCvecGen y$
by replacing each ballot outside of $\GRzVerts$ with the root of
its $\SubGraph i$-subgraph. Therefore, for the profile $\widehat{\LOCvec}\in\GRzVerts^{n}$
the coalition $C$ can, by voting $\widehat{\REPORTset}$, get an
outcome $\MECH_{\GRzVerts}\left(\widehat{\REPORTset},\widehat{\LOCvec}_{-C}\right)$
which it strictly prefers over $\MECH_{\GRzVerts}\left(\widehat{\LOCvec}\right)$,
in contradiction to $\MECH_{\GRzVerts}$ satisfying~$\ThePropSymbol$.\hfill\qedsymbol{}\end{adjustwidth}
\end{proof}%

%\newpage{}

\section{Summary \& Future Work}

In this work, we presented a new family of graphs, \ZVlineGraphsTEXT,
 and a generic anonymous Pareto optimal manipulation-resistant mechanism
for the facility location problem on these graphs.  To the best of
our knowledge, the (very few) false-name-proof mechanisms which are
currently known are for specific graphs and this work is the first
to show a generic false-name-proof mechanism for a large family, utilizing
a  broad graph property and unifying all existence results which
we are aware of.  The construction of the mechanism is  inductive:
We  derive a mechanism for a given \ZVlineGraphTEXT from mechanisms
for its subgraphs. Hence, it is straightforward to derive from our
construction general mechanisms for recursive graph families.

Two technical assumptions we had are connectivity of the graph and
finiteness of the number of agents and locations. Our results can be extended to the case of an infinite number of agents
and locations under common natural constraints like finite diameter
of the graph, measurability of $\NeighOf v$ and of coalitions, and
the order over $\Zvertices$ being a well-order.   It is also not hard to see that the following extension for graphs
in which the connected components are \ZVlineGraphsTEXT will satisfy
the same desiderata.
\begin{itemize}
\item[$\blacktriangleright$] At the first stage, choose the first connected component according
to some predefined order s.t. at least one agent voted for a location
in this component.
\item[$\blacktriangleright$] At the second stage, run our mechanism taking into account only agents
who voted for locations in the chosen component.
\end{itemize}

The mechanism we presented is not the only mechanism satisfying the
desired properties. Taking any other order over the $\Zvertices$
s.t. the constraints of Def.~\ref{def:ZV-line} hold and defining
$\TheMECH$ accordingly will also satisfy them. In particular, a mechanism
which takes at the second stage of Def.~\ref{def:ZV-mechanism} the
rightmost Pareto optimal $\Zvertex$ will also satisfy the same desiderata.
 We did not find any mechanism satisfying the desiderata which is
not of this template, and we conjecture that  these are the only
anonymous Pareto optimal manipulation-resistant mechanisms for facility
location on a graph. 

\begin{conjecture}

Let $\GR=\left(\GRvNEW,\GRe\right)$ be a \ZVlineGraphTEXT w.r.t.
$\WRTaPartition$ and let $\MECH\colon\VtoTheStar\rightarrow\GRvNEW$
be a mechanism s.t.
\begin{itemize}
\item $\MECH$ is anonymous and Pareto optimal; and
\item For any vector of locations $\LOCvec\in\GRvNEW^{n}$,  a coalition
of agents $C$,  and a set of ballots $\REPORTset\in\VtoTheStar$,
 $\REPORTset$ is not a beneficial deviation for $C$.
\end{itemize}
Then, for $i=1,\ldots,k$: Whenever $\LOCvec\in\left(\SubGraph i\right)^{n}$,
 also $\MECHof{\LOCvec}\in\SubGraph i$. Moreover,  $\MECH$ is
the outcome of applying Def.~\ref{def:ZV-mechanism} for some order
over $\GRzVerts$ which satisfies the constraints of Def.~\ref{def:ZV-line}
and mechanisms $\MECHi i$ which are defined by $\LOCvec\in\left(\SubGraph i\right)^{n}\mapsto\MECHof{\LOCvec}$.
\end{conjecture}

Furthermore, unifying non-existence results for specific graphs we've
found  so far, we think that the partition to $\Zvertices$ and $\Vvertices$
is  a fundamental property of a false-name-proof mechanism. Consequentially,
showing that a given graph does not have such structure could be an
easy and efficient way to prove non-existence of a desired mechanism.
\begin{conjecture}
\label{conj:WithTheAlmost}

For almost any graph $\GR=\left\langle \GRvNEW,\GRe\right\rangle $,
if there exists an anonymous and Pareto optimal mechanism $\MECH\colon\VtoTheStar\rightarrow\GRvNEW$
s.t. 
\begin{itemize}
\item[] For any vector of locations $\LOCvec\in\GRvNEW^{n}$,  a coalition
of agents $C$,  and a set of ballots $\REPORTset\in\VtoTheStar$,
 $\REPORTset$ is not a beneficial deviation for $C$.
\end{itemize}
then there exists a sequence of non-empty sets of vertices $\GRzVerts,\SubGraph 1,\ldots,\SubGraph k\allowbreak\subseteq\allowbreak\GRvNEW$
s.t. $\GR$ is a \ZVlineGraphTEXT w.r.t. $\WRTaPartition$.
\end{conjecture}

The only counter example we've found to the conjecture is the cycle
of size $5$  $\FiveCycleInner{c}{@R=0pt@C=1em}{3em}{!}{\VERTEX}{\VERTEX}{\VERTEX}{\VERTEX}{\VERTEX}$
(and graphs derived from it, e.g.,  $\FiveCycleDecorated{c}{@R=0pt@C=.6em}{5em}{!}$).
It is not hard to verify that \textbf{\textbullet} a mechanism which
returns the first Pareto optimal location according to one of the
following orders -  \global\long\def\FiveCycleOrder#1#2#3#4#5{\FiveCycleInner{c}{@R=0pt@C=1em}{4em}{!}{#1}{#2}{#3}{#4}{#5}}
 $\FiveCycleOrder 12453$, $\FiveCycleOrder 12543$, $\FiveCycleOrder 12534$,
or their rotations and reflections - is a manipulation-resistant mechanism
and that \textbf{\textbullet} while all these mechanisms are of the
template of Def.~\ref{def:ZV-mechanism} (for all vertices being
$\Zvertices$), these representations do not satisfy the connectivity
constraints of Def.~\ref{def:ZV-line} and the cycle of size $5$
is not a \ZVlineGraphTEXT. We conjecture that this is a representative
extreme exception and intend to characterize the exception and replace
 `almost' in Conjecture~\ref{conj:WithTheAlmost} with an exact
statement.

Last, an  important continuation of this work is analyzing the implications
for \emph{approximate mechanism design without money}~\cite{DBLP:journals/teco/ProcacciaT13}.
That is, assuming the agents are accurately represented by a cost
function (e.g., the distance to the facility or a monotone function
of the distance) and analyzing implications of manipulation-resistance
on the approximability of the minimization problem of natural social
cost functions, e.g., the average cost (Harsanyi's social welfare),
the geometric mean of the costs (Nash's social welfare), or the maximal
cost (Rawls' criterion). For instance, assuming the two conjectures
above, one gets that whenever there is a large disagreement in the
population (i.e., the agents are dispersed over many $\SubGraph i$-subgraphs)
an extreme status-quo alternative must be chosen by the mechanism,
which results in a bad \emph{price of false-name-proofness}. Nowadays,
many aggregation mechanisms are highly susceptible  to double voting
and to false-name manipulations in general (e.g., mechanisms over
huge anonymous networks like the internet, but also other scenarios
in which vote frauds are known to be easy). We think that such results
should open a discussion on the costs of these protocols (since the
benefits are clear).%

\newpage{}

\bibliographystyle{plain}
\providecommand{\noopsort}[1]{}

\end{document}